\documentclass[conference]{IEEEtran}
\IEEEoverridecommandlockouts
\usepackage{cite}
\usepackage{amsfonts}
\usepackage{amsthm}
\usepackage{amssymb}
\usepackage{algorithmic}
\usepackage{graphicx}
\usepackage{textcomp}
\usepackage{xcolor}
\usepackage{algorithm}
\usepackage{multirow}
\usepackage{amsmath}
\usepackage{subcaption}
\usepackage{placeins} %
\usepackage{bm}
\usepackage{makecell}
\usepackage{pgfplots}
\usepackage{enumitem}
\usepackage{balance}
\usetikzlibrary{positioning}

\newcommand{\ourscheme}{ByITFL}

\newcommand{\TotalUser}{\ensuremath{n}}
\newcommand{\Byzantine}{\ensuremath{b}}
\newcommand{\Colluding}{\ensuremath{t}}
\newcommand{\Dropout}{\ensuremath{p}}
\newcommand{\partition}{\ensuremath{m}}

\newcommand{\Dataset}{\ensuremath{D}}

\newcommand{\model}{\ensuremath{\bm{w}}}

\newcommand{\dimension}{\ensuremath{d}}

\newcommand{\globalIteration}{\ensuremath{g}}
\newcommand{\localIteration}{\ensuremath{c}}
\newcommand{\learningRate}{\ensuremath{\eta}}

\newcommand{\trustscore}{\ensuremath{\mathrm{TS}}}
\newcommand{\modelupdate}{\ensuremath{\mathbf{g}}}
\newcommand{\realnormalmodelupdate}{\ensuremath{\tilde{\mathbf{g}}}}
\newcommand{\normalmodelupdate}{\ensuremath{\bar{\mathbf{g}}}}

\newcommand{\sumone}{\ensuremath{\Sigma_1}}
\newcommand{\sumtwo}{\ensuremath{\mathbf{\Sigma}_2}}

\newcommand{\random}{\ensuremath{\lambda}}

\renewcommand{\baselinestretch}{.971}

\newcommand{\approdegree}{\ensuremath{k}}

\DeclareMathOperator{\ReLU}{ReLU}
\DeclareMathOperator{\AGG}{AGG}

\newtheorem{thm}{Theorem}

\newtheorem{proposition}{Proposition}

\def\BibTeX{{\rm B\kern-.05em{\sc i\kern-.025em b}\kern-.08em
    T\kern-.1667em\lower.7ex\hbox{E}\kern-.125emX}}
\begin{document}

\title{Byzantine-Resilient Secure Aggregation for Federated Learning Without Privacy Compromises \\
\thanks{This project has received funding from the German Research Foundation (DFG) under Grant Agreement Nos. BI 2492/1-1 and WA 3907/7-1.}
}

\author{\IEEEauthorblockN{Yue Xia, Christoph Hofmeister, Maximilian Egger, Rawad Bitar}
\IEEEauthorblockA{
School of Computation, Information and Technology, Technical University of Munich, Munich, Germany \\
\{yue1.xia, christoph.hofmeister, maximilian.egger, rawad.bitar\}@tum.de}
}
\maketitle

\begin{abstract}
Federated learning (FL) shows great promise in large scale machine learning, but brings new risks in terms of privacy and security. We propose ByITFL, a novel scheme for FL that provides resilience against Byzantine users while keeping the users' data private from the federator and private from other users. The scheme builds on the preexisting non-private FLTrust scheme, which tolerates malicious users through trust scores (TS) that attenuate or amplify the users' gradients. The trust scores are based on the ReLU function, which we approximate by a polynomial. The distributed and privacy-preserving computation in ByITFL is designed using a combination of Lagrange coded computing, verifiable secret sharing and re-randomization steps.
ByITFL is the first Byzantine resilient scheme for FL with full information-theoretic privacy.

\end{abstract}


\section{Introduction}
Federated learning (FL), first proposed in~\cite{mcmahan2017communication}, is an emerging machine learning paradigm
that allows users to train a model under the coordination of a central entity called federator while keeping their private data local. The training is iterative. 
Per iteration, the federator sends the current global model to the users, who update it based on their local training data, and return the result is sent back to the federator.
Using an aggregation rule, the federator combines the users' updates into the new global model. The process repeats until the model attains certain converge criteria.

FL addresses the privacy concerns of traditional centralized machine learning by transmitting users' local model updates instead of their private data directly. However, the local model updates, i.e. gradients or weights, may contain sensitive information about the users' data. Upon receiving the local updates, the federator can perform gradient inversion attacks \cite{zhu2019deep} \cite{geiping2020inverting} to reconstruct the users' private training data. Private aggregation protocols \cite{bonawitz2017practical} \cite{aono2017privacy} are introduced to ensure that the federator obtains an aggregate of the model updates without revealing any additional information about individual updates.

Beyond privacy, FL gives rise to security concerns caused by users %
sending corrupt updates. We call malicious users with arbitrary behavior Byzantine. Using simple linear aggregation as in FedAvg \cite{mcmahan2017communication}, even a single Byzantine worker can force arbitrary aggregation results \cite{blanchard2017machine}; leading the model to sub-optimal solutions or divergence. The primary countermeasure is removing outliers in the users' local updates. %
Many Byzantine-resilient aggregations have been proposed \cite{blanchard2017machine,yin2018byzantine,guerraoui2018hidden,cao2020fltrust,zhao2022fedinv}, in which the federator removes outliers by observing the individual updates, thereby compromising the users' privacy.

This highlights an inherent tension between Byzantine-resilience and privacy. While the former requires access to the individual local updates to prune outliers, the latter requires concealing the individual local updates. BREA~\cite{so2020byzantine} and ByzSecAgg~\cite{jahani2023byzantine} address this problem by utilizing secret sharing schemes to make Krum~\cite{blanchard2017machine}, a distance-based Byzantine-resilient aggregation rule, privacy-preserving. Both schemes are computationally private and leak the pairwise distances between local updates to the federator.
The encoding polynomial of the pairwise distances in BREA is further no longer completely random; hence, the federator may obtain extra knowledge (cf. Section~\ref{sec:scheme} for more details). ByzSecAgg \cite{jahani2023byzantine} alleviates this problem by additional re-randomization. Other Byzantine-resilient secure aggregation schemes \cite{he2020secure,velicheti2021secure,hao2021efficient,ma2022shieldfl,xhemrishi2023fedgt} are either based on clustering and compromise privacy or require two federators. %
Currently, there is no Byzantine-resilient scheme without privacy compromises in the literature.

For this reason, we propose \ourscheme\ (pronounced “byte FL”), a Byzantine-resilient
and information-theoretically (IT) private secure aggregation scheme. For Byzantine-resilience we build on FLTrust \cite{cao2020fltrust}. The federator collects a small root dataset to compute a federator model update and computes a trust score (TS) for each user based on the relative direction of the user's and federator's model update. 
FLTrust uses the rectified linear unit function $\ReLU(x)=\max(0,x)$ in computing the TSs. We approximate $\ReLU(x)$ by a polynomial of degree $\approdegree$ to enable an IT private protocol. Each user's model update is embedded into a finite field through a stochastic quantizer, partitioned into $\partition$ sub-vectors, secret shared with all users by Lagrange Coded Computing (LCC) \cite{yu2019lagrange}, and verified against corruptions using an IT verifiable secret sharing (ITVSS) scheme \cite{BGW}, \cite{asharov2017full}. 
Re-randomization \cite{gennaro1998simplified} is required before reconstructing the aggregation result to perfectly hide the local model updates. %
From the perspective of LCC, Byzantine users can be seen as errors and dropouts as erasures in a Reed-Solomon code, cf. \cite{mceliece1981sharing}.
The federator decodes the aggregate and updates the global model.
With $\approdegree$ the degree of the approximation polynomial and $\partition$ the number of sub-vectors, the proposed scheme is resilient against any $\Byzantine$ Byzantine users, IT private against the federator and against any $\Colluding$ colluding users, and robust against any $\Dropout$ dropouts, as long as $\TotalUser \geq 2\Byzantine+(\approdegree+2)\cdot (\partition+\Colluding-1)+\Dropout+1$. 

\section{System Model and Preliminaries}

We use $[n]$ to denote the set of positive integers $\{1,\cdots,n\}$ and use $\lfloor x \rfloor$ for the largest integer less than or equal to $x$. All vectors are denoted in bold type and scalars are denoted in normal type. $I(X;Y)$ denotes the mutual information between random variables $X$ and $Y$. $H(X)$ denotes the entropy of $X$.

\subsection{System Model}
We consider FL with a semi-honest federator and $\TotalUser$ users, including $\Byzantine$ Byzantine users, $\Colluding$ colluding users and $\Dropout$ dropout users, as is illustrated in Fig.~\ref{fig:fl}. Each user $i$ holds a private local dataset $\Dataset_i$, in addition, we require the federator to collect a small dataset $\Dataset_0$, called root dataset. 
The federator possesses a $\dimension$-dimensional global model $\model \in \mathbb{R}^\dimension$ and coordinates the training process. The goal is to train the global model using private data held by the users to find the optimal global model $\model^*$ as the solution to the optimization problem $\model^* = \operatorname{argmin}_{\model}F(\model),$ 
where $F(\model)$ is the loss function. $\nabla f(\Dataset, \model)$ denotes an unbiased estimator of the true gradient $\nabla F(\model)$, i.e. $\nabla F(\model)=\mathbb{E}_{\Dataset \sim \mathcal{D}}[\nabla f(\Dataset,\model)]$. This is done iteratively. Specifically, in each global iteration $\globalIteration$, the federator broadcasts the current global model $\model^{(\globalIteration)}$ to all users. Each user $i$, for $i \in [n]$, initializes its local model to the current global model, i.e., $\model^{(0)}_i=\model^{(\globalIteration)}$, and updates it for one or more local iterations based on its local dataset $\Dataset_i$
\begin{equation}\label{eq:localtraining}
    \model^{(\localIteration+1)}_i = \model^{(\localIteration)}_i-\learningRate_u \cdot \nabla f(\Dataset_i; \model_i^{(\localIteration)}), 
\end{equation}
where $\learningRate_u$ is the local learning rate and $\localIteration$ is the local iteration. Upon finishing all local training iterations, users send the local model update 
$ %
    \modelupdate_i=\model_i-\model^{(\globalIteration)}
$ %
to the federator. Meanwhile, the federator trains on a small dataset to get $\modelupdate_0$, which is assumed to be public. Upon receiving the local updates, the federator aggregates them according to some aggregation rule $\AGG$, i.e., $\modelupdate = \AGG(\modelupdate_0,\modelupdate_1, \cdots, \modelupdate_n),$ 
and computes the global model for the next iteration \footnote{
Although we present our scheme in the simple gradient descent setting, it does not depend on the exact update rule and is applicable to, e.g., momentum and higher order methods and adaptive learning rate schedules. Similarly, it is compatible with additional privacy mechanisms based on adding noise to updates, like differential privacy~\cite{dwork2006differential}.}
$ %
    \model^{(\globalIteration+1)} = \model^{(\globalIteration)}-\learningRate \cdot \modelupdate,
$ %
where $\learningRate$ is the global learning rate. 

\subsection{Threat Model and Defense Goals}
We focus on curious entities with unlimited computing resources, requiring perfect IT privacy. We consider Byzantine attackers. They may arbitrarily deviate from the protocol
and have access to all users' datassets.
\ourscheme\ should be resilient even when up to $\Byzantine$ such users collaboratively misbehave. 

Our scheme guarantees the privacy of honest users' local model updates in each iteration. Up to $\Colluding$ users may collude and share information with each other to guess the private data of honest users. The federator is honest-but-curious, thus honestly conducts the protocol but tries to infer as much sensitive information as possible.

Therefore, for each training iteration, the privacy constraint for \ourscheme\ guarantees that any group of up to $\Colluding$ colluding users, once knowing the current global model and computing their local model updates, which are required for the learning task, and knowing their local datasets, which may give information about what the datasets of the honest users look like, cannot learn any additional information about the local model updates of the other honest users:
\begin{equation}\label{eq:privacy_colluding}
    I(\modelupdate_{[n] \setminus \mathcal{T}}; M_{\mathcal{T}} \mid \modelupdate_{\mathcal{T}}, \Dataset_{\mathcal{T}}, \model^{(\globalIteration)}) = 0,
\end{equation}
and, knowing the current global model, the federator local update and the root dataset, the federator should not gain any information about the local model updates of the honest users beyond the aggregation:
\begin{equation}\label{eq:privacy_federator}
    I(\modelupdate_{[n] \setminus \mathcal{T}}; M_f \mid \modelupdate_0, \Dataset_0, \model^{(\globalIteration)}, \modelupdate) = 0, 
\end{equation}
where $\mathcal{T}$ is the set of colluding users, $\left | \mathcal{T} \right |=\Colluding$, $M_{\mathcal{T}}$ describe the messages received by the colluding parties and $M_f$ denotes the intermediate messages received by the federator.

We consider the possibility of a subset of up to $\Dropout$ users experiencing delays or dropping out during protocol. The protocol should be IT private against the curious-but-honest federator and against any collusion of up to $\Colluding$ users, robust against $\Byzantine$ Byzantine users and at the same time be able to tolerate up to $\Dropout$ users staying silent during the execution.

\begin{figure}[!t]
    \vspace{.05cm}
    \centering
    \input{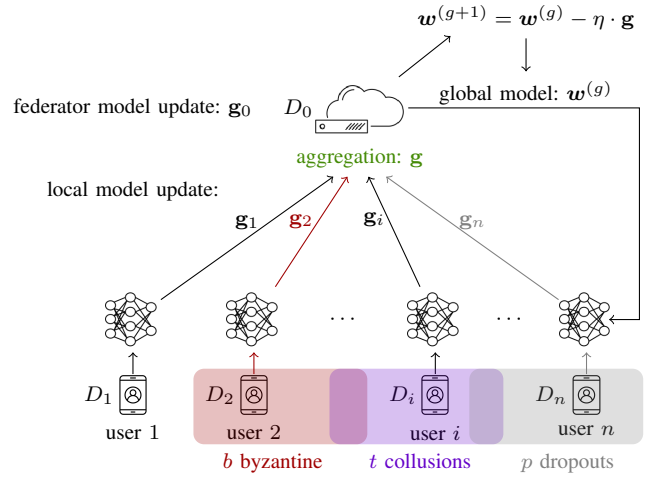}
    \caption{Federated learning system. \vspace{-.5cm}}
    \label{fig:fl}
\end{figure}

\section{\ourscheme}\label{sec:scheme}
We present \ourscheme, which leverages the ability of FLTrust \cite{cao2020fltrust} to provide resilience against Byzantine attacks. Assuming the federator holds a small training dataset and performs local training to obtain a federator model update, \ourscheme\ approximates the ReLU function used to compute the trust scores in FLTrust by a polynomial and uses LCC to provide IT privacy against eavesdroppers. \ourscheme\ consists of the following five main steps, which we will detail in the sequel:
\begin{enumerate}[label=\Alph*.]
    \item Users \textbf{normalize and quantize} their local model updates. The federator model update is treated accordingly.
    \item The \textbf{normalized updates} are partitioned into smaller sub-vectors and \textbf{secret shared} using LCC and ITVSS. %
    \item Users \textbf{validate the normalization} based on the secret shares from other users.
    \item Users \textbf{compute a secret representation} of the aggregation result by evaluating the target polynomials.
    \item The federator receives shares of the aggregation from the users to \textbf{reconstruct the secure aggregation} by decoding an error correcting code %
    and updates the global model.
\end{enumerate}

\subsection{Normalization and Quantization}
To defend against Byzantine attacks performed on the magnitude, the federator and all users first normalize their model update $\modelupdate_i$ to a unit vector
$ %
    \realnormalmodelupdate_i= \frac{\modelupdate_i}{\| \modelupdate_i \|},  \forall i \in \{0,1,\cdots,\TotalUser\},
$ %
so that the impact of extremely large/small local updates, more likely originating from Byzantine users, can be eliminated.

Since the training process is performed in the real domain and LCC (like every IT private secret sharing) works over finite fields, it is essential to transfer the normalized model updates $\realnormalmodelupdate_i \in \mathbb{R}^{\dimension}$ to vectors in a prime field $\normalmodelupdate_i \in \mathbb{F}_p^\dimension$, where $p$ is a large prime. Therefore, users apply an element-wise stochastic quantizer $Q_q(x)$ with $2q+1$ quantization intervals as in \cite{so2020byzantine}, \cite{jahani2023byzantine}. The relation between $p$ and $q$ is explained later. %
Note that the stochastic rounding is unbiased, i.e., $\mathbb{E}_Q[Q_q(x)]=x$. Let $\phi(x) = x+p \mod p$ map integers to values in $\mathbb{F}_p$, the quantization is defined to be
\begin{equation}\label{eq:def_g_bar}
    \normalmodelupdate_i := \phi(q\cdot Q_q(\realnormalmodelupdate_i)).
\end{equation}

\subsection{Sharing of the Normalized Model Updates}
The federator and the users first partition their normalized model updates $\normalmodelupdate_i$ into $\partition$ smaller subvectors
\begin{equation*}
    \normalmodelupdate_i = [\normalmodelupdate_i^{(1)},\normalmodelupdate_i^{(2)},\cdots,\normalmodelupdate_i^{(\partition)}]^T, \forall i \in \{ 0,1,\cdots,\TotalUser\},
\end{equation*}
where each sub-vector is of size $\frac{\dimension}{\partition}$ and $\partition \in [\frac{\TotalUser-\Dropout-1}{\approdegree+2}-\Byzantine-\Colluding+1]$.

We assume the federator model update is public, the federator broadcasts the sub-vectors of $\normalmodelupdate_0$ to the users. Each user $i$ uses LCC \cite{yu2019lagrange} to secret share $\normalmodelupdate_i$ with all users by the degree-$(\Colluding+m-1)$ encoding polynomial
\begin{equation}
    \begin{aligned}
        \bm{u}_i(z) & = \sum_{j \in [\partition]}\normalmodelupdate_i^{(j)} \cdot \prod_{l \in [\partition+t]\setminus \{j\}} \frac{z-\beta_l}{\beta_j-\beta_l} \\
        & + \sum_{j \in [\Colluding]}\bm{r}_i^{(j)} \cdot \prod_{l \in [\partition+t] \setminus \{\partition+j \}} \frac{z-\beta_l}{\beta_{\partition+j}-\beta_l}, \forall i \in [\TotalUser],
    \end{aligned}
\end{equation}
where $\beta_1, \cdots, \beta_{\partition+t}$ are $\partition+t$ distinct elements from $\mathbb{F}_p$ and $\bm{r}^{(j)}$'s are chosen independently and uniformly at random from $\mathbb{F}_p$. Note that $\bm{u}_i(\beta_1)=\normalmodelupdate_i^{(1)}, \cdots, \bm{u}_i(\beta_{\partition})=\normalmodelupdate_i^{(\partition)}$ and the finite field size should be large enough to avoid any wrap around, which we describe in subsection \ref{sub:aggregation}. Secret shares are computed by evaluating $\bm{u}_i(z)$ at $n$ distinct values $\{ \alpha_l \}_{l \in [n]}$, which are selected from $\mathbb{F}_p$ such that $\{ \alpha_l \}_{l \in [n]} \cap \{ \beta_l \}_{l \in [\partition]}=\varnothing$. Hence, each user $j$ receives a secret share of $\normalmodelupdate_i$ from other user $i$, i.e. $\normalmodelupdate_{i,j}=\bm{u}_i(\alpha_j)$, which is a vector of size $\frac{\dimension}{\partition}$, for $i,j \in [\TotalUser]$.

Note that, we leverage the ITVSS protocol from \cite{BGW} to prevent Byzantine users from misbehaving in the secret sharing phase.

\begin{table*}[tb] %
\vspace{.15cm}
\caption{Complexity Analysis with respect to $\TotalUser$, $\dimension$ and $\partition$. Here $\TotalUser$ is the total number of users, $\dimension$ is the dimension of the model updates and the partitioning parameter $\partition \in [\frac{\TotalUser-\Dropout-1}{\approdegree+2}-\Byzantine-\Colluding+1]$. \vspace{-.1cm}}\label{tab:complexity}
\centering
\begin{subtable}{0.4\textwidth} %
\centering
\caption{Communication Complexity}
\begin{tabular}{c|c|c}\label{tab:comm}
    & Per-User & Federator\\ \hline
BREA \cite{so2020byzantine} & $O(dn+n^2)$ & $O(dn+n^3)$ \\ \hline
ByzSecAgg \cite{jahani2023byzantine} & $O(\frac{d}{\partition}n+n^2)$ & $O(\frac{d}{\partition}n+n^3)$ \\ \hline
\ourscheme & $O(\frac{d}{m}n^3+n^4)$ & $O(\frac{d}{m}n+n^2)$ 
\end{tabular}
\end{subtable}%
\begin{subtable}{0.6\textwidth} %
\centering
\caption{Computation Complexity}
\begin{tabular}{c|c|c}\label{tab:comp}
    & Per-User & Federator \\ \hline
BREA \cite{so2020byzantine} & $O(dn\log^2n+dn^2)$ & $O((dn+n^3)\log^2n\log\log n)$ \\ \hline
ByzSecAgg \cite{jahani2023byzantine} & $O(\frac{d}{\partition}n\log^2n+\frac{d}{\partition}n^2)$ & $O((\frac{d}{\partition}n+n^3)\log^2n\log\log n)$ \\ \hline
\ourscheme & \makecell{
$O((\frac{d}{\partition}n^2+n^3)\log^2n\log\log n$ \\ $+ (\frac{d}{\partition}n^3+n^4)\log^2n)$ \\
} & $O((\frac{d}{m}n+n^2)\log^2n\log\log n)$
\end{tabular}
\end{subtable}
\vspace{-0.3cm}
\end{table*}

\subsection{Validation of Normalization}\label{sub:normalization}
Malicious users may misbehave during normalization. 
Thus, upon receiving a secret share, each user $i$ verifies correct normalization by locally computing the squared $l2$-norm $\| \normalmodelupdate_{j,i} \|_2^2$ of the secret shares  for each $j\in [\TotalUser]$ and sending the computed shares to the federator. 
This is possible due to LCC and a re-randomization step before sending computations to the federator, which will be detailed in the next subsection. Upon receiving the computation results, the federator utilizes error correction decoding of the underlying Reed-Solomon code, cf. \cite{mceliece1981sharing}, to reconstruct $\| \normalmodelupdate_i \|_2^2$ for each $\normalmodelupdate_i$ and checks if it is within a certain interval, i.e.,
\begin{equation}\label{eq:checknormalization}
    \left |   \| \normalmodelupdate_i \|_2^2 - \phi(q\cdot Q_q(1))^2 \right | < \varepsilon \cdot q^2,
\end{equation}
where $\varepsilon$ is a predefined threshold and can be set empirically. %
Note that the error correction requires the total number of users $\TotalUser \geq 2\Byzantine+2(\partition+\Colluding-1)+\Dropout+1$. The interval is caused by the accuracy loss due to quantization. If any user does not pass the normalization check, the federator marks them as Byzantine and excludes them from future computations.

\subsection{Users Secure Computation}\label{sub:computation}
In FLTrust \cite{cao2020fltrust}, the federator assigns to each user a trust score 
$ %
    \trustscore_i=\ReLU(\cos(\theta_i)), \forall i \in [\TotalUser], 
$ %
where $\theta_i$ is the angle between the federator's and the user's model update. The federator aggregates the local model updates by averaging the normalized updates weighted by their trust scores. 

Making FLTrust IT private is not straightforward, which is why we propose to approximate the ReLU by a degree-$\approdegree$ polynomial function $h(x)=h_0+h_1x+ \cdots + h_kx^{\approdegree}$.
Therefore, the trust score for each user becomes
\begin{equation*}%
    \trustscore_i \approx h(\cos(\theta_i))=h(\langle \normalmodelupdate_0,\normalmodelupdate_i \rangle), \forall i \in [\TotalUser], 
\end{equation*}
and the aggregation result is
\begin{equation}\label{eq:aggre}
\begin{aligned}
    \modelupdate & =\frac{\| \modelupdate_0 \|}{\sum_{i\in [\TotalUser]}{\trustscore_i}} \cdot \sum_{i\in [\TotalUser]}{(\trustscore_i \cdot \normalmodelupdate_i)}
    = \frac{\| \modelupdate_0 \|}{\sumone} \cdot \sumtwo, 
\end{aligned}
\end{equation}
\begin{equation*}
\begin{aligned}
    \text{where } \sumone &= \sum_{i\in [\TotalUser]}{h(\langle \normalmodelupdate_0,\normalmodelupdate_i \rangle)} \text{ and }
    \sumtwo = \sum_{i\in [\TotalUser]}{(h(\langle \normalmodelupdate_0,\normalmodelupdate_i \rangle) \cdot \normalmodelupdate_i)}.
\end{aligned}
\end{equation*}

Since the federator possesses $\modelupdate_0$, for computing the aggregation $\modelupdate$ in a privacy-preserving manner, the federator only needs to compute the value of $\frac{\sumtwo}{\sumone}$ in a private manner without learning individual users' private information beyond this quotient. The colluding users should learn nothing about other honest users during the computation.

\textit{Privacy Against Colluding Users:}
Both $\sumone$ and $\sumtwo$ are polynomial functions of the model updates $\normalmodelupdate_0$ and $\normalmodelupdate_i$ for $i \in [n]$, where $\normalmodelupdate_0$ is public and $\normalmodelupdate_i$'s are secret shared among the users using LCC. It's worth mentioning that LCC allows the computation of an arbitrary polynomial $f$ with degree $\deg(f)$ over its secret. Suppose user $i$ holds a secret $\bm{s}_i$, the user partitions it and shares it among users via a degree-$(\partition+\Colluding-1)$ encoding polynomial $\bm{u}_i(z)$. Each user $j$, holding its secret share $\bm{s}_{i,j}=\bm{u}_i(\alpha_j)$, is able to compute $f(\bm{s}_{i,j})=f(\bm{u}_i(\alpha_j))$ locally, which is an evaluation of the resulting polynomial $f(\bm{u}_i(z))$ at the point $\alpha_j$. Upon having more than $(m+t-1) \deg(f) + 1$ evaluations from the users, the resulting polynomial $f(\bm{u}_i(z))$ can be interpolated. The desired computation is obtained by evaluating the polynomial $f(\bm{u}_i(z))$ at points $\{ \beta_l \}_{l \in [\partition]}$, i.e., $f(\bm{s}_i)=[f(\bm{s}_i^{(1)}),\cdots, f(\bm{s}_i^{(\partition)})]^T=[f(\bm{u}_i(\beta_1)),\cdots, f(\bm{u}_i(\beta_\partition))]^T$. 

Hence, it is possible to perform the polynomial computations of $\sumone$ and $\sumtwo$ on the secret shares, such that each user obtains an evaluation point of $\sumone$ and $\sumtwo$. This guarantees that any set of up to $\Colluding$ users are not able to learn anything from the shares. 

\textit{Privacy Against the Federator:}
Privacy against the federator has not yet been guaranteed: the reconstructions of the computation results, i.e. $\| \normalmodelupdate_i \|_2^2$, $\sumone$ and $\sumtwo$, cannot perfectly hide the secret values $\modelupdate_1, \cdots, \modelupdate_\TotalUser$ against the federator. LCC is a linear secret sharing scheme that is additively, but not multiplicatively, homomorphic. Given two secrets $a$ and $b$ shared among $n$ users with encoding polynomial $u_a(z)$ and $u_b(z)$, each user $i$, having the secret shares $a_i=u_a(\alpha_i)$ and $b_i=u_b(\alpha_i)$, is able to locally compute the sum of the shares $a_i+b_i=u_a(\alpha_i)+u_b(\alpha_i)=u_{a+b}(\alpha_i)$, which perfectly hides the secrets. %
This property does not hold for multiplication on shares \cite{BGW,asharov2017full}. The product of $a_i$ and $b_i$ results in a secret share of $u_a(z)\cdot u_b(z)$, whose evaluation at $\beta_1$ is indeed $a \cdot b$, 
but is not a completely random polynomial perfectly hiding the secret, i.e. $u_a(z)\cdot u_b(z) \neq u_{a\cdot b}(z)$. %
The federator can learn additional information about $a$ and $b$ beyond $a\cdot b$. 
We follow the re-randomization step from \cite{gennaro1998simplified,asharov2017full}, which involves sub-sharing the users' secret shares using ITVSS \cite{BGW}, and linearly combining to construct re-randomized secret shares.

The users own the re-randomized shares of $\sumone$ and $\sumtwo$. It remains to ensure that the federator obtains the quotient $\sumtwo/\sumone$ without gaining any additional information about $\sumone$ and $\sumtwo$. 
Therefore, each user $i$ 
1) chooses an independent value $\random_i$ uniformly at random from $\mathbb{F}_p$ and secret shares it by LCC and ITVSS among all users,
2) adds the shares of $\random_j$'s from all other user $j$ and obtains the share of\footnote{The case $\random=0$ can be avoided by minor changes, omitted for brevity.}  $\random = \sum_{j=[\TotalUser]}{\random_j}$. 
Each user multiplies the re-randomized shares of $\sumone$ and $\sumtwo$ by their share of $\random$, performs another re-randomization and sends the resulting shares of $\random \sumone$ and $\random \sumtwo$ to the federator.

\subsection{Secure Aggregation}\label{sub:aggregation}
The federator receives the secret shares of $\random \sumone$ and $\random \sumtwo$, where the degree of the encoding polynomial for
$\random \sumone$ is $(\approdegree+1) \cdot (\partition+\Colluding-1)$ and for $\random \sumtwo$ is $(\approdegree+2)\cdot (\partition+\Colluding-1)$. With sufficient number of users sending evaluations to the federator, the federator is able to leverage the error correction property of Reed-Solomon codes \cite{mceliece1981sharing} to decode the values of $\random \sumone$ and $\random \sumtwo$. Therefore, we require the total number of users in the system to be $\TotalUser \geq 2\Byzantine+(\approdegree+2)\cdot (\partition+\Colluding-1)+\Dropout+1$. 

Upon decoding the correct values, the federator computes $\sumtwo/\sumone = \random\sumtwo/\random\sumone$, converts the results from the finite field back to the real domain by de-quantizing by $Q_q(x)^{-1}$ and demapping by $\phi^{-1}$ %
and computes the global model for the next iteration. To ensure the correctness of the result, none of the computations should cause a wrap around in the finite field. Each entry of the normalized gradient is in the range $-q$ to $q$, hence the dot product is in the range $-dq^2$ to $dq^2$. Accounting for the $0$ value, we thus require Thus, we require $p \geq 2\TotalUser \dimension^\approdegree q^{2\approdegree+1}+1$.

\section{Theoretical Analysis}
We analyze the properties of our proposed scheme in theory, beginning with the privacy guarantee achieved by \ourscheme.

\begin{thm}\label{theorem_1} 
\ourscheme\ with $\TotalUser \geq 2\Byzantine+(\approdegree+2)\cdot (\partition+\Colluding-1)+\Dropout+1$ guarantees IT privacy of the honest users' local model updates against any number $\Colluding$ of colluding users or the federator according to \eqref{eq:privacy_colluding} and \eqref{eq:privacy_federator}. %
\end{thm}

\begin{proof}
We first prove the privacy against any $\Colluding$ colluding users according to \eqref{eq:privacy_colluding}:
\begin{align}
    & I(\modelupdate_{[n] \setminus \mathcal{T}}; M_{\mathcal{T}} \mid \modelupdate_{\mathcal{T}}, \Dataset_{\mathcal{T}}, \model^{(\globalIteration)}) \nonumber \\
    & = H(\modelupdate_{[n] \setminus \mathcal{T}} | \modelupdate_{\mathcal{T}}, \Dataset_{\mathcal{T}}, \model^{(\globalIteration)}) - H(\modelupdate_{[n] \setminus \mathcal{T}} | M_{\mathcal{T}}, \modelupdate_{\mathcal{T}}, \Dataset_{\mathcal{T}}, \model^{(\globalIteration)}) \nonumber \\
    & = H(\modelupdate_{[n] \setminus \mathcal{T}}, \Dataset_{\mathcal{T}}, \model^{(\globalIteration)}) - H(\Dataset_{\mathcal{T}}, \model^{(\globalIteration)}) \nonumber \\
    &\quad  - H(\modelupdate_{[n] \setminus \mathcal{T}}, M_{\mathcal{T}}, \Dataset_{\mathcal{T}}, \model^{(\globalIteration)}) + H(M_{\mathcal{T}}, \Dataset_{\mathcal{T}}, \model^{(\globalIteration)}), \label{eq:proof_colluding}
\end{align}
where the last equation follows because $\modelupdate_{\mathcal{T}}$ is a deterministic function of $\Dataset_{\mathcal{T}}$ and $\model^{(\globalIteration)}$. We then consider the exchanged messages $M_{\mathcal{T}}$, which include the shares of the normalized local model updates $\normalmodelupdate_{\mathcal{T},i}$ and sub-shares from the re-randomization step in Step C and Step D of the scheme, i.e. when computing the squared $l2$-norm $\| \normalmodelupdate_i \|_2^2$ to check the correctness of the normalization and the two sums $\sumone$ and $\sumtwo$ in the aggregation. We can leverage the privacy guarantees of LCC \cite{yu2019lagrange}, the re-randomization step \cite{gennaro1998simplified}, \cite{asharov2017full} and ITVSS \cite{BGW}. When the number of colluding users $\Colluding$ satisfies $\TotalUser \geq 2\Byzantine+(\approdegree+2)\cdot (\partition+\Colluding-1)+\Dropout+1$, the exchanged messages $M_{\mathcal{T}}$ observed by the colluding users are completely random and independent of $\modelupdate_{[n] \setminus \mathcal{T}}$, $\Dataset_{\mathcal{T}}$ and $\model^{(\globalIteration)}$, i.e., 
\begin{equation*}
\begin{aligned}\label{eq:proof_colluding_equa}
    H(\modelupdate_{[n] \setminus \mathcal{T}}, M_{\mathcal{T}}, \Dataset_{\mathcal{T}}, \model^{(\globalIteration)}) 
    & = H(\modelupdate_{[n] \setminus \mathcal{T}}, \Dataset_{\mathcal{T}}, \model^{(\globalIteration)}) + H(M_{\mathcal{T}}), \\
    H(M_{\mathcal{T}}, \Dataset_{\mathcal{T}}, \model^{(\globalIteration)})
    & = H(\Dataset_{\mathcal{T}}, \model^{(\globalIteration)}) + H(M_{\mathcal{T}}).
\end{aligned}
\end{equation*}
By substituting the above into \eqref{eq:proof_colluding} we prove \eqref{eq:privacy_colluding}, showing that \ourscheme\ is IT private against $\Colluding$ users. 
For privacy against the honest-but-curious federator, we need to prove \eqref{eq:privacy_federator}. We have 
\begin{equation*}\label{eq:proof_federator}
\begin{aligned}
    & I(\modelupdate_{[n] \setminus \mathcal{T}}; M_f \mid \modelupdate_0, \Dataset_0, \model^{(\globalIteration)}, \modelupdate) \\
    & = H(M_f \mid \modelupdate_0, \Dataset_0, \model^{(\globalIteration)}\!\!, \modelupdate) \! - \! H(M_f \mid \modelupdate_{[n] \setminus \mathcal{T}}, \modelupdate_0, \Dataset_0, \model^{(\globalIteration)}\!\!, \modelupdate) \\
\end{aligned} 
\end{equation*}
With regard to the exchanged messages $M_f$, we need to consider: 1) the computed shares of $\| \normalmodelupdate_i \|_2^2$, $\random\sumone$ and $\random\sumtwo$ sent from the users, which are completely random, i.e. uniformly distributed and independent of $\modelupdate_{[n] \setminus \mathcal{T}}, \modelupdate_0, \Dataset_0, \model^{(\globalIteration)}, \modelupdate$, by leveraging the privacy guarantee of the re-randomization step \cite{gennaro1998simplified}, \cite{asharov2017full}, and 2) the reconstructed values of $\| \normalmodelupdate_i \|_2^2$, $\random\sumone$ and $\random\sumtwo$. Since $\| \normalmodelupdate_i \|_2^2$ lies within a certain range for all possible model updates (ideally equivalent to one), the value of $\| \normalmodelupdate_i \|_2^2$ is also independent of $\modelupdate_{[n] \setminus \mathcal{T}}, \modelupdate_0, \Dataset_0, \model^{(\globalIteration)}, \modelupdate$. With respect to $\random\sumone$ and $\random\sumtwo$, $\sumone$ is completely hidden due to the randomness $\random$, but $\random\sumtwo$ may leak the relationships within $\sumtwo$. However, by computing $\sumtwo/\sumone$, which is equivalent to knowing $\random\sumone$ and $\random\sumtwo$, the interrelationships within $\sumtwo$ has to be leaked. Note that the value of $\sumtwo/\sumone$ does not result in any additional information leakage because $H(\sumtwo/\sumone|\modelupdate, \modelupdate_0)=0$, as a result of $\modelupdate= \| \modelupdate_0 \| \cdot \sumtwo / \sumone$. We denote the computed shares as $\bm{c}$, and have
\begin{align*}
    & H(M_f \mid \modelupdate_0, \Dataset_0, \model^{(\globalIteration)}, \modelupdate) \nonumber \\
    &= H(\bm{c}, \| \normalmodelupdate_i \|_2^2, \random\sumone, \random\sumtwo \mid \modelupdate_0, \Dataset_0, \model^{(\globalIteration)}, \modelupdate) \nonumber \\
    &= H(\bm{c}, \| \normalmodelupdate_i \|_2^2, \sumtwo/\sumone \mid \modelupdate_0, \Dataset_0, \model^{(\globalIteration)}, \modelupdate) \nonumber \\
    & = H(\bm{c}, \| \normalmodelupdate_i \|_2^2 \mid \modelupdate_0, \Dataset_0, \model^{(\globalIteration)}, \modelupdate) 
    = H(\bm{c}, \| \normalmodelupdate_i \|_2^2),
\end{align*}
and $H(M_f \mid \modelupdate_{[n] \setminus \mathcal{T}}, \modelupdate_0, \Dataset_0, \model^{(\globalIteration)}, \modelupdate)= H(\bm{c}, \| \normalmodelupdate_i \|_2^2)$. This concludes the proof of IT privacy against the federator. Hence, we conclude the proof of Theorem 1. 
\end{proof}
\begin{figure}[!t]
    \vspace{.05cm}
    \centering
    \input{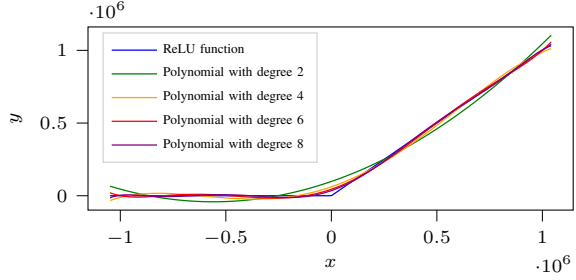}
    \caption{Comparison of the polynomial-approximated ReLU with different degrees, where $y=\ReLU(x)$ for ReLU function and $y=h(x)$ for the approximations. \vspace{-.5cm}}
    \label{fig:relu_approximation}
\end{figure}

\begin{proposition}
    \ourscheme\ requires a user to communicate $O(\frac{\dimension}{\partition}n^3+n^4)$ and the federator $O(\frac{d}{m}n+n^2)$ scalars. The computation cost is $O((\frac{d}{\partition}n^2+n^3)\log^2n\log\log n + (\frac{d}{\partition}n^3+n^4)\log^2n)$ and $O((\frac{d}{m}n+n^2)\log^2n\log\log n)$, respectively.
\end{proposition}
\begin{proof}[Sketch of Proof]
For each user sharing a single scalar, the computation complexity of encoding in LCC is $O(n\log^2 n\log\log n)$ \cite{yu2019lagrange} and that of ITVSS \cite{BGW} is $O(n^2 \log^2 n)$. Since the re-randomization step \cite{gennaro1998simplified} involves sub-sharing each secret share using VSS and a linear combination of the sub-shares, the computation complexity is $O(n^2\log^2 n\log\log n+n^3 \log^2 n)$. The communication cost for the ITVSS is $O(n^2)$, and $O(n^3)$ for re-randomization. With respect to the federator, the computation cost for decoding a single scalar using error-correction decoding takes $O(n\log^2 n\log\log n)$. The remaining proof follows counting arguments, here omitted for brevity.
\begin{figure}[!t]
    \vspace{.03cm}
    \centering
    \input{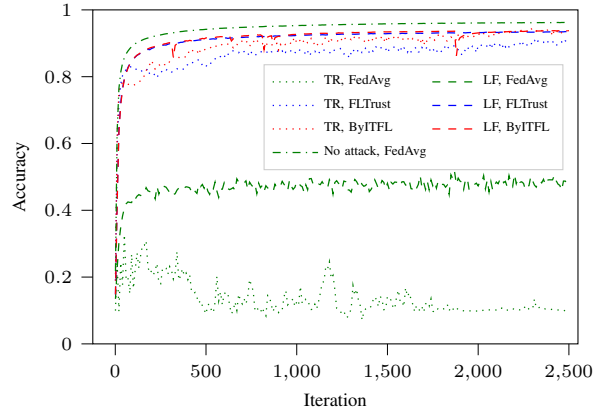}
    \caption{Test accuracy of \ourscheme\ vs. FedAvg \cite{mcmahan2017communication} and FLTrust \cite{cao2020fltrust} under MNIST, uniform data distribution, trim attack (TR) and label flipping (LF) \cite{fang2020local}. The baseline is FedAvg. \vspace{-.5cm}}
    \label{fig:test_accuracy_attack}
\end{figure}
\end{proof}

In Table \ref{tab:complexity}, we compare the communication and computation complexity of \ourscheme\ with respect to $\TotalUser$, $\dimension$ and $\partition$ to the previous solutions BREA and ByzSecAgg. In \ourscheme, the large communication and computation complexity for users stems from the ITVSS scheme and the re-randomization. This is the cost of achieving IT privacy, while BREA and ByzSecAgg rely on computational privacy. By partitioning the updates into $\partition$ sub-vectors, the complexities can be reduced.

\section{Experiments}
We numerically evaluate the convergence and Byzantine-resilience of \ourscheme\ using the approximated ReLU function, and demonstrate it compared to FedAvg and FLTrust. Note that, since we have theoretically established the privacy guarantee, we focus on the performance of the approximated ReLU function against Byzantine attacks in the experiments, which is required by our privacy guarantee.
We consider $10$-class image classification on MNIST with uniform data distribution across $\TotalUser=40$ users and a three-layer neural network. %
$50\%$ of the users are Byzantine ($\Byzantine=20$) and perform either a trim attack, which is an untargeted local model poisoning attack presented in \cite{fang2020local}, or a label flipping attack with the same setting as  \cite{fang2020local}.
We set $q=1024$ and $\approdegree=6$. As in \cite{cao2020fltrust}, the size of the root dataset $\Dataset_0$ is $100$. %
We set $\varepsilon=0.02$ for the normalization check. 
The ReLU function is approximated between $(-q^2, q^2)=(-1024^2, 1024^2)$, which are the values taken by the cosine similarity after quantization. Fig.~\ref{fig:relu_approximation} shows the ReLU approximation with different polynomial degrees $\approdegree\in \{2,4,6,8\}$. %
Fig.~\ref{fig:test_accuracy_attack} shows the convergence of \ourscheme\ compared to FedAvg and FLTrust under the trim attack and label flipping. Notice that FedAvg cannot defend against Byzantine attacks, and \ourscheme\ has a comparable performance to FLTrust.

\section{Conclusion}
We proposed a Byzantine-resilient secure aggregation scheme for FL that guarantees IT privacy. \ourscheme\ requires the federator to hold a root dataset for reference, used to scale the users' local model updates during the aggregation. To enable IT privacy, we used a suitable approximation of ReLU to compute trust scores and secret sharing techniques to ensure user privacy. We analyzed the achieved privacy and complexity of the algorithm. Through experiments, we demonstrated the convergence in the presence of Byzantine users. Extending the established methods to wireless settings such as in \cite{egger2023private} is an interesting direction. Further, we consider the investigation of targeted attacks tailored to \ourscheme\ as important future work.

\balance

\bibliographystyle{IEEEtran}
\bibliography{IEEEabrv,refs}

\end{document}